\documentclass[authoryear,a4paper, 12pt]{elsarticle}
\usepackage{natbib}
\usepackage{amsthm}
\usepackage{amssymb}
\usepackage{graphicx,amsmath}
\usepackage{tikz}
\usepackage{mathtools}
\usepackage{natbib}
\usepackage{booktabs,xcolor}
\usepackage{graphics}
\usepackage{enumitem}
\usepackage{subcaption}
\usepackage{appendix}
\usepackage{adjustbox}
\usetikzlibrary{patterns}
\usepackage{caption}
\usepackage{setspace}
\usepackage[colorlinks=true,linkcolor=blue, urlcolor=blue, citecolor=blue, breaklinks]{hyperref}
\journal{   }

\newtheorem{proposition}{Proposition}

\theoremstyle{definition}

\usepackage{etoolbox}
\makeatletter
\patchcmd{\ps@pprintTitle}
{Preprint submitted to}
{ }
{}{}
\makeatother
\patchcmd{\emailauthor}{(#2)}{}{}{}
\patchcmd{\urlauthor}{(#2)}{}{}{}

\usepackage{tikz}

\newcommand{\XX}{\mathcal{X}}

\begin{document}

\begin{frontmatter}
\title{Fairness and Efficiency in Cake-Cutting with Single-Peaked Preferences}
\author[add1]{Bhavook Bhardwaj}
\ead{bhavook17r@isid.ac.in}
\author[add2]{Rajnish Kumar}
\ead{rajnish.kumar@qub.ac.uk}
\author[add2]{Josu\'e Ortega}
\ead{j.ortega@qub.ac.uk}

\address[add1]{Indian Statistical Institute, Delhi, India.}
\address[add2]{Queen's Management School, Queen's University Belfast, UK.}

\date{\today}

\begin{abstract}
	We study the cake-cutting problem where agents have single-peaked preferences over the cake. We show that a recently proposed mechanism by \cite{wang2019cake} to obtain envy-free allocations can yield large welfare losses. Using a simplifying assumption, we characterize all Pareto optimal allocations, which have a simple structure: are peak-preserving and non-wasteful. Finally, we provide simple alternative mechanisms that Pareto dominate that of Wang--Wu and achieve envy-freeness or Pareto optimality.
\end{abstract}

\end{frontmatter}
	
	\newpage
	\setcounter{footnote}{0}

\section{Introduction}
\label{sec:introduction}

The problem of how to fairly divide a heterogeneous good among $n$ agents is one of the oldest problems in fair division \citep{steinhaus1948}. Known as cake-cutting, this problem has spanned a long interdisciplinary literature in economics, computer science, mathematics and political science that includes multiple real-life applications, such as the division of land, inheritances, and computing resources. 

One of the  main goals of the cake-cutting literature is the construction of division methods that produce envy-free allocations. Such methods exist, but they require a large amount of information regarding agents' preferences \citep{brams1995envy,aziz2016discrete}. However, when agents' preferences are well-behaved, an envy-free allocation can be computed with  few queries about agents' preferences. In particular, when agents have single-peaked preferences, an envy-free allocation can be achieved with only $2n$ queries, using a clever procedure recently proposed by  \cite{wang2019cake}.\footnote{In the more general class of piecewise linear preferences, an envy-free allocation can be computed using $\Omega(n^6)$ queries \citep{kurokawa2013cut}. A connected envy-free allocation cannot be obtained by a mechanism with finite queries \citep{stromquist2008envy}.} Single-peaked preferences are particularly interesting to study for two reasons. First, they can be described concisely: agents are specified as peak points, who prefer the pieces of cake that are close to their own locations. Second, single-peaked preferences can realistically describe preferences over land resources, clean water and mineral deposits; and have been thoroughly studied in other fair division problems.\footnote{E.g. \cite{moulin1984generalized,sprumont1991division, ehlers2002strategy,kasajima2013probabilistic,long2019strategy}, among others. Single-peaked preferences have also been (implicitly) used in cake-cutting \citep[][p. 17]{weller1985fair}.}

Although the Wang--Wu procedure achieves envy-freeness with low informational requirements, it may produce allocations that are Pareto dominated. In this note, we show that the Wang--Wu procedure may waste up to, but not more than, $\frac{n-1}{n}$ of the total achievable utility (Proposition \ref{thm:prop1}). Obtaining a Pareto optimal allocation is easy only when single-peaked preferences have a common slope; in this case they have a simple structure and are completely characterized by being peak-preserving and non-wasteful (Proposition \ref{thm:prop2}). Under those assumptions, we can construct two simple mechanisms that Pareto dominate the Wang--Wu mechanism, while attaining envy-freeness or Pareto optimality. Alternatively, the utilitarian mechanism that returns the allocation that maximizes the sum of agents' utilities is another reasonable option that is efficient and also easy to describe (Proposition \ref{thm:prop3}, Table \ref{tab:properties}).

\section{Model}
We consider the standard cake-cutting problem \citep{nicolo2008,procaccia2016}. The cake is the $[0, 1]$ interval and the set of agents is $N=\{1, . . . , n\}$. A piece of cake $P$ is a finite set of disjoint subintervals of $[0, 1]$. Each agent $i\in N$ has a integrable, non-negative value density function $v_i$. For each piece of cake $P$, an agent’s value or utility on $P$ is denoted by $V_i(P)$ and defined by the integral of its density function, i.e., $V_i (P) \coloneqq \int_{P}v_i(x) dx$. Thus, $V_i : 2^{[0,1]} \rightarrow \mathbb{R}^+$ is called the valuation or utility function of agent $i$. The definition of the valuation function implies that it is additive and non-atomic. We simplify notation by writing $V_i(x, y)$ instead of $V_i([x, y])$ for each interval $[x, y]$. We normalize agents’ valuations so that $V_i(0, 1) = 1$ for each $i\in N$. A {\it cake-cutting problem} is a triple $([0,1], N, (v_i)_{i \in N})$.

We assume that agents' preferences over the cake are single-peaked. Agent $i$ has a \emph{single-peaked valuation} when her value density function $v_i$ satisfies that there is a number $p_i \in [0, 1]$ such that $v_i(p_i) = \max_{x \in [0,1]} v_i(x)$ and $v_i(x) = \max\{0, v_i(p_i) -k_i\left\vert x-p_i \right\vert\}$ for all $x \in [0, 1]$ and some coefficient $k_i > 0$. We call $p_i$ and $v_i(p_i)$ the peak location and peak density of agent $i$, respectively. Notice that the coefficient $k_i$ as well as the density function is uniquely determined by the peak location and peak density as the normalized assumption says $\int_{0}^{1} v_i(x)dx = 1$. We use the notation $U_i$ to denote the maximum subinterval such that agent $i$ has positive value on every point, i.e. $v_i(x)>0$ for all $x \in U_i$ and $v_i(x)=0$ for all $x \in [0,1]\setminus U_i$. The points $l_i$ and $r_i$ denote the left and right endpoints of $U_i$. We assume without loss of generality that $\bigcup_{i \in N}  U_i=[0,1]$.\footnote{Pieces of cake that are not desired by anybody can be dropped or assigned arbitrarily, and play no role with regards to fairness and efficiency considerations.}

An \emph{allocation} $X = (X_1, . . . , X_n)$ is a partition of the cake among the $n$ agents, i.e. each agent $i$ receives the piece of cake $X_i$, the pieces are disjoint, and $\bigcup_{i \in N} X_i = [0,1] $. An allocation $X$ is \emph{proportional} if $V_i(X_i) \geq 1/n$ for all $i \in N,$ and \emph{envy-free} (EF) if $V_i(X_i) \geq V_i(X_j)$ for all $i, j \in N$. Envy-freeness implies proportionality, and they are equivalent for $n=2$. None of these fairness properties, however, imply efficiency in the sense of Pareto optimality. An allocation $X$ is \emph{Pareto optimal} (PO) if there is no other allocation $X'$ that Pareto dominates it, i.e. such that, for all $i \in N$, $V_i(X'_i) \geq V_i(X_i)$; and for some $i$, $V_i(X'_i) > V_i(X_i)$. 

A mechanism $M$ is a function mapping a series of queries to agents about their valuations into an allocation. In the standard \cite{robertson1998} framework, in which we focus, only two types of queries are allowed: either an agent is asked to \emph{cut} the cake at an specific point so that the left piece has a specific valuation, or an agent is asked to \emph{evaluate} a piece of cake, i.e. reveal her valuation for such a piece. Mechanism $M$ is said to be envy-free (resp. proportional, Pareto optimal) if it always produces an envy-free (resp. proportional, Pareto optimal) allocation. 

\section{Results}

\cite{wang2019cake} provide an ingenious mechanism to obtain envy-free allocations with only $2n$ queries. Their mechanism asks two cut queries to each agent (in the standard Robertson--Webb model) to find the $l_i, p_i, r_i$ points (see their Algorithm 1). Then, all those points are ordered in increasing sequence, together with the $0$ and $1$ endpoints. Between any two of these points, the cake is divided into $2n$ equidistant pieces, and agent 1 receives pieces $1$ and $2n$, agent 2 receives pieces $2$ and $2n - 1$, and so on until agent $n$, who receives pieces $n$ and $n+1$ (see Figure \ref{fig:fig1} for an example). Because valuations are single-peaked, their procedure computes an envy-free allocation with only $2n$ queries (i.e. substantially less queries than those required in more general preference domains, see \cite{kurokawa2013cut}).

However, a problem with the Wang--Wu mechanism is that it may produce an allocation that is not Pareto optimal, as they point out. In fact, the Wang--Wu procedure only produces a Pareto optimal allocation in a very special case: when $p_i=p_j$ and $k_i=k_j$ for all $i,j \in N$. 

Moreover, the welfare losses than can occur in the Wang--Wu mechanism can be substantially large, as the example in Figure \ref{fig:fig1} shows. In our example, the Wang--Wu procedure gives to each agent only $1/n$ of the utility that they obtain in the unique PO allocation. The example can be extended to any number of agents to show the severity of the welfare losses in the Wang--Wu procedure \textit{vis-a-vis} Pareto optimal allocations. To formalise this observation, let $\XX^{PO}$ be the set of all PO allocations and $X^{WW}$ be the allocation suggested by the Wang--Wu procedure. Let $T^{PO} \coloneqq \displaystyle \min_{X \in \XX^{PO}} \frac{\sum_{i \in N} V(X_i)}{n}$ and $T^{WW} \coloneqq \frac{\sum_{i \in N} V(X^{WW}_i)}{n}$.\footnote{The choice of the Pareto optimal allocation with the minimum sum of utilities is arbitrary. Because in the proof we construct a cake-cutting problem with a unique Pareto optimal allocation, our result can be rephrased using $T^{PO} \coloneqq \displaystyle \max_{X \in \XX^{PO}}\frac{\sum_{i \in N} V(X_i)}{n}$. Our definition of welfare losses is related to the \emph{price of fairness} \citep{bertsimas2011price}.} The welfare losses in the Wang--Wu mechanism are defined as $WL \coloneqq T^{PO}-T^{WW}$. Our first Proposition provides a tight upper bound for $WL$.

\begin{figure}[h]
\centering
\begin{subfigure}{.99\linewidth}
	\centering
	\begin{tikzpicture}
	\draw[->] (-.1,0) -- (6.2,0) node[right] {$x$};
	\draw[->] (0,-.1) -- (0,3.2) node[above] {$v_i(x)$};
	
	\draw[] (0,0)--(1,3) -- (2,0)-- (0,0);
	\draw[] (6,0)--(5,3) -- (4,0)-- (6,0);
	\draw[] (4,0)--(3,3) -- (2,0)-- (4,0);
	
	\node[] at (1,-.3) {$p_1$};
	\node[] at (3,-.3) {$p_2$};
	\node[] at (5,-.3) {$p_3$};
	\node[] at (6,-.3) {$1$};
	\end{tikzpicture}
	\caption{Single-peaked preferences of three agents over the cake.} \label{fig:example1}
\end{subfigure}\\

\begin{subfigure}{.5\linewidth}
	\centering
	\begin{tikzpicture}[scale=.65]
	\draw[->] (-.1,0) -- (6.2,0);
	\draw[->] (0,-.1) -- (0,3.2);

	\draw[fill=gray] (0,0)--(1,3) -- (2,0)-- (0,0);
	\draw[fill=black] (6,0)--(5,3) -- (4,0)-- (6,0);
	\draw[] (4,0)--(3,3) -- (2,0)-- (4,0);
	
	\end{tikzpicture}
	\caption{Unique Pareto-optimal allocation.} \label{fig:M1}
\end{subfigure}%
\begin{subfigure}{.5\linewidth}
	\centering
	\begin{tikzpicture}[scale=.65]
	\draw[->] (-.1,0) -- (6.2,0) ;
	\draw[->] (0,-.1) -- (0,3.2) ;
	
	\draw[] (0,0)--(1,3) -- (2,0)-- (0,0);
	\draw[] (6,0)--(5,3) -- (4,0)-- (6,0);
	\draw[] (4,0)--(3,3) -- (2,0)-- (4,0);
	
	\draw[] (1,0)--(1,3);
	\draw[] (3,0)--(3,3);
	\draw[] (5,0)--(5,3);
	
	\foreach \x in {1,2,3,4,5}
	\draw (\x/6,0) -- (\x/6,\x*.5);
	\foreach \x in {1,2,3,4,5}
	\draw (2-\x/6,0) -- (2-\x/6,\x*.5);
	
	\foreach \x in {1,2,3,4,5}
	\draw (2+\x/6,0) -- (2+\x/6,\x*.5);
	\foreach \x in {1,2,3,4,5}
	\draw (4-\x/6,0) -- (4-\x/6,\x*.5);
	
	\foreach \x in {1,2,3,4,5}
	\draw (4+\x/6,0) -- (4+\x/6,\x*.5);
	\foreach \x in {1,2,3,4,5}
	\draw (6-\x/6,0) -- (6-\x/6,\x*.5);
	
	\draw[fill=gray] (0,0)--(1/6,0) -- (1/6,.5)-- (0,0);
	\draw[fill=gray] (1,0)--(5/6,0) -- (5/6,2.5)-- (1,3) -- cycle;
	\draw[fill=gray] (2,0)--(2+1/6,0) -- (2+1/6,.5)-- (2,0);
	\draw[fill=gray] (4,0)--(4+1/6,0) -- (4+1/6,.5)-- (4,0);
	\draw[fill=gray] (3,0)--(2+5/6,0) -- (2+5/6,2.5)-- (3,3) -- cycle;
	\draw[fill=gray] (5,0)--(4+5/6,0) -- (4+5/6,2.5)-- (5,3) -- cycle;
	
	\draw[fill=gray] (2,0)--(2-1/6,0) -- (2-1/6,.5)-- (2,0);
	\draw[fill=gray] (4,0)--(4-1/6,0) -- (4-1/6,.5)-- (4,0);
	\draw[fill=gray] (6,0)--(6-1/6,0) -- (6-1/6,.5)-- (6,0);
	
	\draw[fill=gray] (1,3)--(1+1/6,2.5) -- (1+1/6,0)-- (1,0) -- cycle;
	\draw[fill=gray] (3,3)--(3+1/6,2.5) -- (3+1/6,0)-- (3,0) -- cycle;
	\draw[fill=gray] (5,3)--(5+1/6,2.5) -- (5+1/6,0)-- (5,0) -- cycle;
	
	\draw[fill=black] (2/6,1)--(4/6,2) -- (4/6,0)-- (2/6,0) -- cycle;
	\draw[fill=black] (2+2/6,1)--(2+4/6,2) -- (2+4/6,0)-- (2+2/6,0) -- cycle;
	\draw[fill=black] (4+2/6,1)--(4+4/6,2) -- (4+4/6,0)-- (4+2/6,0) -- cycle;
	
	\draw[fill=black] (2-2/6,1)--(2-4/6,2) -- (2-4/6,0)-- (2-2/6,0) -- cycle;
	\draw[fill=black] (4-2/6,1)--(4-4/6,2) -- (4-4/6,0)-- (4-2/6,0) -- cycle;
	\draw[fill=black] (6-2/6,1)--(6-4/6,2) -- (6-4/6,0)-- (6-2/6,0) -- cycle;
	\end{tikzpicture}
	\caption{Wang-Wu allocation} \label{fig:M2}
\end{subfigure}
\caption{Agents 1, 2, and 3 receive the grey, white, and black pieces respectively. In the unique PO allocation (which is also EF), the sum of individual valuations is 3, whereas in the Wang--Wu allocation, the sum is only 1.}
\label{fig:fig1}
\end{figure}
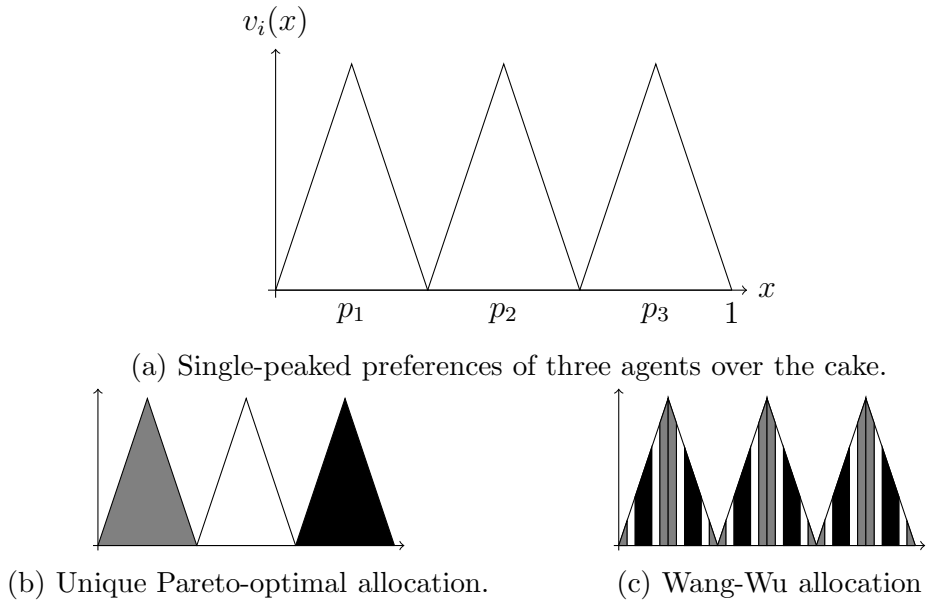

\begin{proposition}
	\label{thm:prop1}
	For any $n$, the welfare losses in the Wang--Wu mechanism $WL$ can be up to, but no more than, $\frac{n-1}{n}$.
\end{proposition}

\begin{proof}
	To show that the upper bound of $\frac{n-1}{n}$ is achievable, construct a cake-cutting problem such that $\bigcap_{i\in N} U_i = \emptyset$. Then, there is a unique Pareto optimal allocation $X^{PO}$ such that $V_i(X_i^{PO})=1$ for all $i \in N$. 	In contrast, in the Wang--Wu procedure $[l_i,p_i]$ is divided into $2n$ equidistant pieces, of which agent $i$ receives 2 that she jointly values at $1/2n$; similarly she receives 2 of the $2n$ equidistant pieces in which the interval $[p_i, r_i]$ is divided, which she also jointly values at $1/2n$. Thus, $V_i(X_i^{WW})=1/n$ for all $i \in N$, and thus $T^{PO}-T^{WW}=\frac{n-1}{n}$.

	To show that this upper bound is tight, note that the Wang--Wu procedure is envy-free, and thus proportional, so that $V_i(X_i^{WW})\geq 1/n$, and thus, for any cake-cutting problem, $T^{PO}-T^{WW} \leq \frac{n-1}{n}$.
\end{proof}

The trade-off between efficiency and envy-freeness is not specific to the Wang--Wu procedure, and is present in several fair division problems \citep[e.g.][]{maniquet2000resource}, including cake-cutting.\footnote{Large welfare losses also arise in matching problems \citep{ortega2018social,ortega2019losses}.} To understand this trade-off better, we provide a characterization of all Pareto-optimal mechanisms. This is substantially easier to do when all valuation density functions $v_i$ have the same slope, i.e. $k_i=k$ for all $i \in N$, an assumption that we impose from now.\footnote{One may consider the following broader preference domain: if $l_i \leq l_j$, then $r_i \leq r_j$ for all $i,j \in N$. Our results do not extend to this more general preference domain.} We focus on the generic case in which all peaks are different, i.e. $p_i \neq p_j$ for all $i, j \in N$.\footnote{Our results extend to the case when some agents have the same peaks \emph{mutatis mutandis}. In such a case, an allocation may be disconnected only between agents who have the same peaks.}

An allocation $X$ is \emph{non-wasteful} if $X_i \subseteq U_i$ for all $i$. A connected allocation is completely described by $n-1$ cutting points $c_1 \leq \ldots \leq c_{n-1}$ such that $X_1=[0,c_1]$, $X_2=(c_1,c_2]$, and so on. A connected allocation is \emph{peak-preserving} if the interval $[0,c_1]$ is allocated to the agent with the smallest peak point, $(c_1,c_2]$ is allocated to the agent with the second smallest peak point, and so on, until the piece $(c_{n-1},1]$ is allocated to the agent with the largest peak point. In our second Proposition, we show that non-wastefulness and peak-preservingness completely characterize Pareto optimal allocations, which implies that all Pareto optimal allocations are connected.

\begin{proposition}
	\label{thm:prop2}
	When $k_i=k$ for all $i\in N$, an allocation is Pareto optimal if and only if it is a non-wasteful and peak-preserving allocation.
\end{proposition}

\begin{proof}
	First, we show (by contradiction) that if $X$ is Pareto optimal, then it must be non-wasteful and peak preserving. If $X$ was wasteful, some agent $i$ would receive a part of the cake that she does not value, but which by assumption is valued by someone else. A clear Pareto improvement exists, and thus $X$ is not Pareto optimal, a contradiction. 
	
	If $X$ was not peak-preserving, it means it is either not connected or it is connected but it does not respect peaks. In both cases, there must exist two agents $i,j$ such that $p_i<p_j$, $[a,b]\subseteq X_j$ and $(b,c]\subseteq X_i$, where $a<b<b'<c$ (see subfigure 2.a).\footnote{That agents $i,j$ as above must exist is trivial when $X$ does not respect peaks. If $X$ is not connected, there is an agent $\alpha$ who obtains at least two disconnected pieces of cake $[\alpha^0,\alpha']$ and $(\alpha'',\alpha''']$, and agents $\beta$ and $\gamma$ with allocations $(\alpha', \beta']$ and $(\gamma', \alpha'']$, such that $\alpha'< \beta'\leq \gamma' <\alpha''$. If $p_\alpha > p_\beta$, then $i=\beta$ and $j=\alpha$. Otherwise, if $p_\alpha < p_\beta$ and $p_\alpha<p_\gamma$, then $i=\alpha$ and $j=\gamma$. And if $p_\alpha < p_\beta$ and $p_\alpha>p_\gamma$, then $i=\gamma$ and $j=\beta$.} We show geometrically in Figure \ref{fig:fig2} that a reassignment of the interval $[a,c]$ such that $i$ receives $[a,b')$ and $j$ receives $[b',c]$ leads to a Pareto improvement, where $b'$ is chosen so that $V_j(b',c)=V_j(a,b)$.\footnote{A similar argument is used to show that, when two agents divide several objects, at most one object is divided in all Pareto optimal allocations \cite[][p. 255]{moulin2004}. See also Corollary 2.4 in \cite{sandomirskiy2019fair}.}

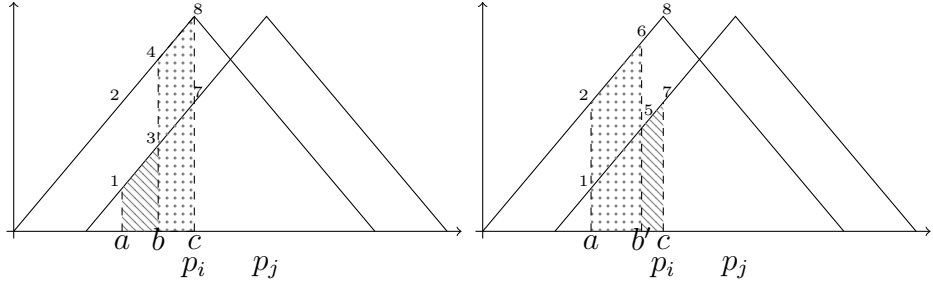
\begin{figure}[h!]
		\centering
\begin{subfigure}{.45\linewidth}
	
	\begin{tikzpicture}[scale=.95]
	\draw[->] (-.1,0) -- (6.2,0);
	\draw[->] (0,-.1) -- (0,3.2);
	
	\draw[-] (0,0) -- (2.5,3) -- (5,0);
	\draw[-] (1,0) -- (3.5,3) -- (6,0);
	
	\draw[pattern=north west lines, pattern color=gray,draw=none] (1.5,0)--(2,0) -- (2,1.2)-- (1.5,.6);
	\draw[pattern=dots, pattern color=gray, distance=.01pt, dashed] (2.5,0) -- (2,0) -- (2,2.4) -- (2.5,3);
	\draw[dashed] (2.5,0) -- (2.5,3);
	\draw[dashed] (1.5,0) -- (1.5,.6);

	\node[] at (2.5,-.5) {$p_i$};
	\node[] at (3.5,-.5) {$p_j$};
	
	\node[] at (1.5,-.15) {$a$};
	
	\node[] at (2,-.1) {$b$};
	\node[] at (2.5,-.15) {$c$};
	
	\node[] at (1.4,.7) {\tiny $1$};
	\node[] at (1.4,1.9) {\tiny $2$};
	\node[] at (1.9,1.3) {\tiny $3$};
	\node[] at (1.9,2.5) {\tiny $4$};
	\node[] at (2.55,1.95) {\tiny $7$};
	\node[] at (2.55,3.1) {\tiny $8$};
	
	\end{tikzpicture}
	\caption{A non peak-preserving allocation.} \label{fig:N1}
\end{subfigure}\begin{subfigure}{.45\linewidth}
	\begin{tikzpicture}[scale=.95]
	\draw[->] (-.1,0) -- (6.2,0);
	\draw[->] (0,-.1) -- (0,3.2);
	
	\draw[-] (0,0) -- (2.5,3) -- (5,0);
	\draw[-] (1,0) -- (3.5,3) -- (6,0);
	
	\draw[pattern=dots, pattern color=gray,draw=none] (1.5,0)--(2.2,0) -- (2.2,2.65)-- (1.5,1.8);
	\draw[pattern=north west lines, pattern color=gray,draw=none] (2.5,0)--(2.2,0) -- (2.2,1.45)-- (2.5,1.8);
	
	\draw[dashed] (1.5,0) -- (1.5,1.8);
	\draw[dashed] (2.5,0) -- (2.5,1.8);
	\draw[dashed] (2.2,0) -- (2.2,2.65);
	
	\node[] at (2.5,-.5) {$p_i$};
	\node[] at (3.5,-.5) {$p_j$};
	\node[] at (1.5,-.15) {$a$};
	\node[] at (2.2,-.1) {$b'$};
	\node[] at (2.5,-.15) {$c$};
	
	\node[] at (1.4,.7) {\tiny $1$};
	\node[] at (1.4,1.9) {\tiny $2$};
	\node[] at (2.3,1.7) {\tiny $5$};
	\node[] at (2.2,2.8) {\tiny $6$};
	\node[] at (2.55,1.95) {\tiny $7$};
	\node[] at (2.55,3.1) {\tiny $8$};
	
	
	\end{tikzpicture}
	\caption{A Pareto improvement.} \label{fig:N2}
\end{subfigure}
		\caption{In subfigure 2.a, we have a non peak-preserving allocation in which agent $j$ gets $[a,b]$ and agent $i$ gets $(b,c]$
			(the sum of utilities is the shaded area). In subfigure 2.b, we show how such allocation is Pareto dominated by a peak-preserving allocation. We reassign the interval $[a,c]$ so that $j$ gets an interval $(b',c]$ such that $V_j(b',c)=V_j(a,b)$. The difference $V_i(a,b')+V_j(b',c)-V_j(a,b)-V_i(b,c)$ is equal to the area of polygon $1-2-6-5$ (coordinates $(a,v_j(a))-(a,v_i(a))-(b',v_i(b'))-(b',v_j(b'))$) minus the area of the polygon $3-4-8-7$ (coordinates $(b,v_j(b))-(b,v_i(b))-(c,v_i(c))-(c,v_j(c))$). The former is strictly larger than the latter because $b'>b$, due to the fact that $c \leq p_j$. Since by construction $V_j(b',c)=V_j(a,b)$, we must have that $V_i(a,b')>V_i(b,c)$, showing that the allocation in subfigure 2.a is Pareto dominated. The argument is completely symmetrical when the allocation endpoints $a,b,c$ are located elsewhere.}
		\label{fig:fig2}
	\end{figure}
	
Second, we prove that if $X$ is non-wasteful and peak-preserving, then it must be Pareto optimal. 
By contradiction, suppose $X$ is Pareto dominated by an alternative allocation $X'$ so that for some agent $i$, $V_i(X_i')>V_i(X_i)$ and $V_j(X'_j) \geq V_j(X_j)$ for all $j \in N$. $X'$ can be assumed to be Pareto optimal, otherwise there exists another allocation $X''$ that Pareto dominates both $X'$ and $X$ (we can reach the Pareto optimal allocation by trading among agents). So let $X'$ be Pareto optimal, which by our previous argument must be non-wasteful and peak preserving. Both $X$ and $X'$ can be completely described by a sequence of $n-1$ cut points, called $c=(c_1,\ldots, c_{n-1})$ and $c'=(c'_1,\ldots, c'_{n-1})$, with endpoint $c_n=c'_n=1$. Let $s$ be the smallest index for which $c'_s<c_s$, if such exists. Then, because $X$ is non-wasteful, $V_s(X_s)>V_s(X'_s)$, and thus $X'$ does not Pareto dominate $X$. If no such index exist, $V_n(X_n)>V_n(X'_n)$, and thus $X'$ does not Pareto dominate $X$. We conclude that $X$ must be Pareto optimal.
\end{proof}

The assumption of common slopes is critical to obtain such a simple characterization of Pareto optimality: when each agent's valuation density has a different slope, Pareto optimal allocations may be disconnected, as Figure \ref{fig:figdisc} shows.

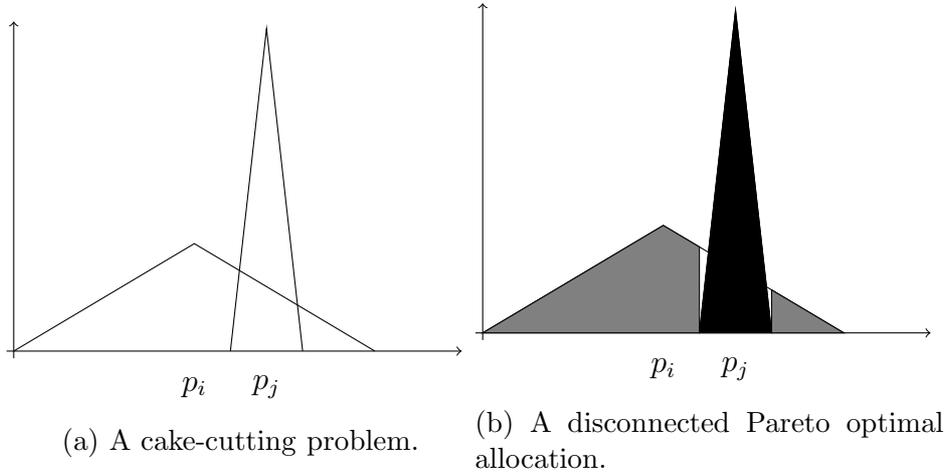
\begin{figure}[h!]
			\centering
	\begin{subfigure}{.45\linewidth}
		
		\begin{tikzpicture}[scale=.95]
		\draw[->] (-.1,0) -- (6.2,0);
		\draw[->] (0,-.1) -- (0,4.6);
		
		\draw[-] (0,0) -- (2.5,1.5) -- (5,0);
		\draw[-] (3,0) -- (3.5,4.5) -- (4,0);

		\node[] at (2.5,-.5) {$p_i$};
		\node[] at (3.5,-.5) {$p_j$};

		\end{tikzpicture}
		\caption{A cake-cutting problem.} \label{fig:N1}
	\end{subfigure}\begin{subfigure}{.45\linewidth}
		\begin{tikzpicture}[scale=.95]
		\draw[->] (-.1,0) -- (6.2,0);
		\draw[->] (0,-.1) -- (0,4.6);
		
		\draw[-] (0,0) -- (2.5,1.5) -- (5,0);
		\draw[-] (3,0) -- (3.5,4.5) -- (4,0);

		\node[] at (2.5,-.5) {$p_i$};
		\node[] at (3.5,-.5) {$p_j$};
		
		\draw[fill=gray] (0 ,0)--(2.5,1.5)--(3,1.2) -- (3,0) -- (0,0) -- cycle;
		\draw[fill=gray] (4 ,0)--(4,.6) -- (5,0) -- cycle;
		\draw[fill=black] (3,0) -- (3.5,4.5) -- (4,0) -- (3,0) -- cycle;
		\end{tikzpicture}
		\caption{A disconnected Pareto optimal allocation.} \label{fig:N2}
	\end{subfigure}
	\caption{An cake-cutting problem with single-peaked valuations and different slopes. The allocation in which agent $i$ obtains the grey parts, and agent $j$ obtains the black parts is Pareto optimal, but not connected.}
	\label{fig:figdisc}
\end{figure}

\section{Three different mechanisms that improve on Wang--Wu}

We turn to the design of a mechanism that returns an allocation that is both Pareto optimal and envy-free. Although such an allocation exists \citep{weller1985fair}, constructing an algorithm that returns it (in the Robertson--Webb model) remains an open question when preferences are single-peaked (and is impossible in more general preference domains \citep{kurokawa2013cut}). Nevertheless, we provide three deterministic mechanisms that achieve both properties individually knowing only agents' peaks and endpoints, and that, at the same time, improve on the original Wang--Wu procedure by either Pareto dominating it or generating a higher total utility. We note that the three mechanisms we describe below are in the Robertson--Webb model, since the peak and endpoints of valuations' densities can be recovered using Algorithm 1 in \cite{wang2019cake}. Figure \ref{fig:fig3} exemplifies how our three different solutions differ and improve on the Wang--Wu mechanism.

\noindent
\begin{enumerate}[leftmargin=0cm,itemindent=.5cm,labelwidth=\itemindent,labelsep=0cm,align=left]

\item {\it Utilitarian Mechanism} (UM). Our first mechanism computes the allocation that maximizes the sum of agents' utilities, i.e. the area under the upper envelope of the valuation density functions. Computing such allocation is easy: it is the peak-preserving allocation described by $n-1$ cut points $c^u_1,\ldots,c^u_{n-1}$, each located at the intersection of the valuation density functions, i.e. $c^u_i$ is such that $v_i(c^u_i)=v_{i+1}(c^u_{i+1})$. The utilitarian mechanism is Pareto optimal (and arguably the best one among all Pareto optimal ones, since it maximizes the sum of utilities), but it may create envy.\footnote{The utilitarian mechanism maximizes the sum of utilities, but maximizing other aggregate utility functions may also be of interest, such as the Nash product \citep{segal2019monotonicity}. See \cite{juarez2013implementing} for further applications of the utilitarian mechanism.} 

\item {\it Leftmost leaves} (LL). Our second mechanism is leftmost leaves and is also simple. Index agents according to their peaks, so that agents 1 is the one with the smallest peak, and let $c^{ll}_1$ be the point for which $V_1[0,c^{ll}_1]=1/n$ and $c'^{ll}_1=max\{c^{ll}_1,l_2\}$. Agent 1 is assigned the interval $[0,c'^{ll}_1]$. Then, let $c^{ll}_2$ be the point such that $V_2[c'^{ll}_1,c^{ll}_2]=\frac{V_2[c'^{ll}_1,1] }{n-1}$ and $c'^{ll}_2=max\{c^{ll}_2,l_3\}$. Agent 2 is assigned the interval $(c'^{ll}_1,c'^{ll}_2]$ and so on, until the last agent who receives the interval $(c'^{ll}_{n-1},1]$. Leftmost leaves generates a Pareto optimal allocation, guaranteeing all agents at least a utility of $1/n$ \citep{kyropoulou2019fair,ortega2019obvious}, in contrast with the Wang--Wu mechanism which guarantees exactly a $1/n$ utility to all agents; thus leftmost leaves Pareto dominates the Wang--Wu mechanism. However, leftmost leaves can generate envy.

\item {\it Modified Wang--Wu} (MWW). The third mechanism is a modification of the original Wang--Wu mechanism. As in the original version, the three points $l_i, p_i, r_i$ for all agents are ordered in increasing sequence, together with the $0$ and $1$ endpoints. The cake is cut at each of these points, generating up to $3n+1$ cake pieces. The difference with the original Wang--Wu procedure is that each of these pieces is divided only among the agents who have a positive value for it. Formally, for each cake piece $P_i$ let $n_i \coloneqq \{i \in N : U_i \bigcap P_i \neq \emptyset\}$. Then, each piece $P_i$ is divided into $2n_i$ equidistant pieces, which are shared into the interested $n_i$ agents as in the Wang--Wu procedure: the first interested agent gets pieces $1$ and $2n_i$, the second one gets pieces $2$ and $2n-1$ and so on. This procedure is envy-free (like the original Wang--Wu procedure, which it Pareto dominates), although is not Pareto optimal since it generates a disconnected allocation. 

We summarize our findings in our third Proposition and in Table \ref{tab:properties}.
\end{enumerate}

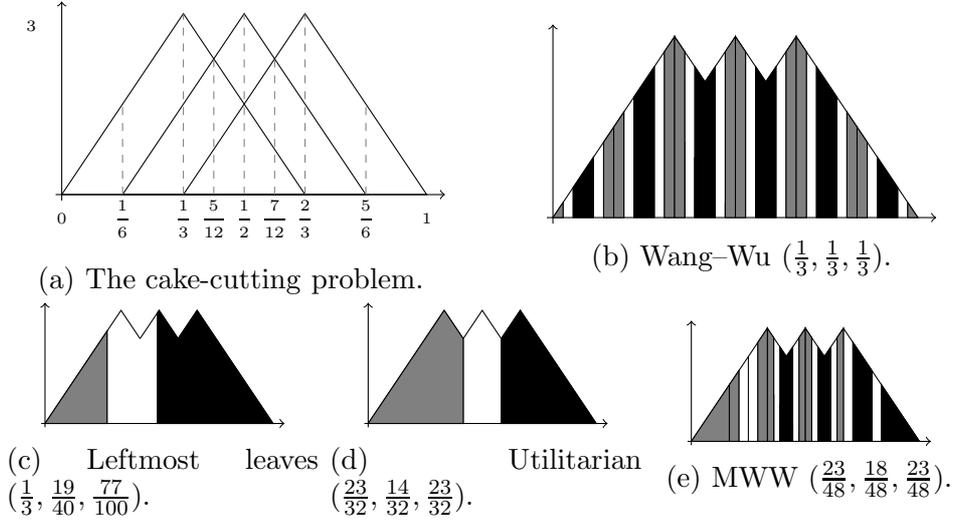
\begin{figure}[h!]
	\centering
	\begin{subfigure}{.49\linewidth}
		\centering
		\begin{tikzpicture}[scale=.8]
		\draw[->] (-.1,0) -- (6.3,0) ;
		\draw[->] (0,-.1) -- (0,3.2) ;
		
		\draw[] (0,0)--(2,3) -- (4,0)-- (0,0);
		\draw[] (1,0)--(3,3) -- (5,0)-- (1,0);
		\draw[] (2,0)--(4,3) -- (6,0)-- (2,0);
		
		\draw[dashed,gray] (1,0)--(1,1.5);
		\draw[dashed,gray] (2,0)--(2,3);
		\draw[dashed,gray] (3,0)--(3,3);
		\draw[dashed,gray] (4,0)--(4,3);
		\draw[dashed,gray] (5,0)--(5,1.5);
		\draw[dashed,gray] (2.5,0)--(2.5,9/4);
		\draw[dashed,gray] (3.5,0)--(3.5,9/4);

		\node[] at (0,-0.4) {\tiny $0$};
		\node[] at (1,-0.4) {\tiny $\dfrac{1}{6}$};
		\node[] at (2,-0.4) {\tiny $\dfrac{1}{3}$};
		\node[] at (2.5,-0.4) {\tiny $\dfrac{5}{12}$};
		\node[] at (3,-0.4) {\tiny $\dfrac{1}{2}$};
		\node[] at (3.5,-0.4) {\tiny $\dfrac{7}{12}$};
		\node[] at (4,-0.4) {\tiny $\dfrac{2}{3}$};
		\node[] at (5,-0.4) {\tiny $\dfrac{5}{6}$};
		\node[] at (-0.5,2.8) {\tiny $3$};
		\node[] at (6, -0.4) {\tiny $1$};
		\end{tikzpicture}
		\caption{The cake-cutting problem.} \label{fig:M1}
	\end{subfigure}%
	\begin{subfigure}{.49 \linewidth}
		\centering
		\begin{tikzpicture}[scale=.8]
		\draw[->] (-.1,0) -- (6.3,0) ;
		\draw[->] (0,-.1) -- (0,3.2) ;
		
		\draw[] (0,0)--(2,3) -- (2.5,2.25);
		\draw[] (2.5,2.25)--(3,3) -- (3.5,2.25);
		\draw[] (3.5,2.25)--(4,3) -- (6,0);
		
		
		\foreach \x in {1,2,3,4,5,6}
		\draw (\x/6,0) -- (\x/6,\x/6*1.5);
		\foreach \x in {1,2,3,4,5}
		\draw (1+\x/6,0) -- (1+\x/6,\x/6*1.5 +1.5);
		
		\foreach \x in {1,2,3}
		\draw (2+\x/6,0) -- (2+\x/6,3 -\x/6*1.5);
		\foreach \x in {4,5}
		\draw (2+\x/6,0) -- (2+\x/6,1.5 +\x/6*1.5);
		\foreach \x in {1,2,3,4,5}
		\draw (3-\x/6,0) -- (3-\x/6,\x/6*1.5);
		
		\foreach \x in {1,2,3}
		\draw (3+\x/6,0) -- (3+\x/6,3 -\x/6*1.5);
		\foreach \x in {4,5}
		\draw (3+\x/6,0) -- (3+\x/6,1.5 +\x/6*1.5);
		\foreach \x in {1,2,3,4,5}
		\draw (3-\x/6,0) -- (3-\x/6,\x/6*1.5);
		
		\foreach \x in {1,2,3,4,5,6}
		\draw (4+\x/6,0) -- (4+\x/6,3 -\x/4);
		\foreach \x in {1,2,3,4,5}
		\draw (5+\x/6,0) -- (5+\x/6,1.5 -\x/4);
		
		\draw[fill=gray] (0 ,0)--(1/6,0) -- (1/6,1/4)-- (0,0) -- cycle;
		\draw[fill=gray] (5/6 ,0)--(1,0) -- (1 , 1.5)-- (5/6,5/4) -- cycle;
		\draw[fill=black] (2/6 ,0)--(4/6,0) -- (4/6 , 1)-- (2/6,1/2) -- cycle;
		\draw[fill=gray] (1 ,0)--(1+1/6,0) -- (1+1/6,7/4)-- (1,1.5) -- cycle;
		\draw[fill=gray] (1+5/6 ,0)--(2,0) -- (2 , 3)-- (1+5/6,11/4) -- cycle;
		\draw[fill=black] (1+2/6 ,0)--(1+4/6,0) -- (1+ 4/6 , 10/4)-- (1+ 2/6,2) -- cycle;
		
		\draw[fill=gray] (2 ,0)--(2+1/6,0) -- (2+1/6,11/4)-- (2,3) -- cycle;
		\draw[fill=gray] (2+5/6 ,0)--(3,0)--(3,3) -- (2+5/6,11/4) -- cycle;
		\draw[fill=black] (2+2/6 ,0)--(2+4/6,0)--(2+4/6,10/4) -- (2+3/6,9/4) -- (2+2/6, 10/4) -- cycle;
		\draw[fill=gray] (3 ,0)--(3+1/6,0) -- (3+1/6,11/4)-- (3,3) -- cycle;
		\draw[fill=gray] (3+5/6 ,0)--(4,0) -- (4,3)-- (3+5/6,11/4) -- cycle;
		\draw[fill=black] (3+2/6 ,0)--(3+4/6,0)--(3+4/6,10/4) -- (3+3/6,9/4) -- (3+2/6, 10/4) -- cycle;
		
		\draw[fill=gray] (4 ,0)--(4+1/6,0) -- (4+1/6,11/4)-- (4,3) -- cycle;
		\draw[fill=gray] (4+5/6 ,0)--(5,0)--(5,1.5) -- (4+5/6,7/4) -- cycle;
		\draw[fill=black] (4+2/6 ,0)--(4+4/6,0)--(4+4/6,2) -- (4+2/6,10/4) -- cycle;
		\draw[fill=gray] (5 ,0)--(5+1/6,0) -- (5+1/6,5/4)-- (5,1.5) -- cycle;
		\draw[fill=gray] (5+5/6 ,0)--(6,0)-- (5+5/6,1/4) -- cycle;
		\draw[fill=black] (5+2/6 ,0)--(5+4/6,0)--(5+4/6,1/2)--(5+2/6,1) -- cycle;
		\end{tikzpicture}
		\caption{Wang--Wu ($\frac{1}{3}, \frac{1}{3}, \frac{1}{3}$).} \label{fig:M2}
	\end{subfigure} %
	
	\begin{subfigure}{.3 \linewidth}
		\centering
		\begin{tikzpicture}[scale=.5]
		\draw[->] (-.1,0) -- (6.3,0) ;
		\draw[->] (0,-.1) -- (0,3.2) ;
		
		\draw[] (0,0)--(2,3) -- (2.5,2.25);
		\draw[] (2.5,2.25)--(3,3) -- (3.5,2.25);
		\draw[] (3.5,2.25)--(4,3) -- (6,0);
		
		\draw[] (0.2721*6,0)--(0.2721*6,0.2721*6*1.5);
		\draw[] (0.4932*6,0)--(0.4932*6,0.3265*6*1.5);
		\draw[fill=gray] (0.2721*6,0)--(0.2721*6,0.2721*6*1.5) -- (0,0) --cycle;
		\draw[fill=black] (0.4932*6,0)--(0.4932*6,0.3265*6*1.5)-- (3,3)-- (3.5,9/4) -- (4,3) --(6,0)--(0.4932*6,0) --cycle;
		
		\end{tikzpicture}
		\caption{Leftmost leaves  ($\frac{1}{3},\frac{19}{40},\frac{77}{100})$.} \label{fig:M3}
	\end{subfigure} %
	\begin{subfigure}{.3\linewidth}
		\centering
		\begin{tikzpicture}[scale=.5]
		\draw[->] (-.1,0) -- (6.3,0) ;
		\draw[->] (0,-.1) -- (0,3.2) ;
		
		\draw[] (0,0)--(2,3) -- (2.5,2.25);
		\draw[] (2.5,2.25)--(3,3) -- (3.5,2.25);
		\draw[] (3.5,2.25)--(4,3) -- (6,0);
		
		\draw[] (2.5,0)--(2.5,9/4);
		\draw[] (3.5,0)--(3.5,9/4);
		\draw[fill=gray] (0 ,0)--(2.5,0)--(2.5,9/4) -- (2,3) -- (0,0) -- cycle;
		\draw[fill=black] (3.5 ,0)--(6,0)--(4,3) -- (3.5,9/4) -- (3.5,0) -- cycle;
		\end{tikzpicture}
		\caption{Utilitarian ($\frac{23}{32},\frac{14}{32},\frac{23}{32}$).} \label{fig:M5}
	\end{subfigure} %
	\begin{subfigure}{.3\linewidth}
		\centering
		\begin{tikzpicture}[scale=.5]
		\draw[->] (-.1,0) -- (6.3,0) ;
		\draw[->] (0,-.1) -- (0,3.2) ;
		
		\draw[] (0,0)--(2,3) -- (2.5,2.25);
		\draw[] (2.5,2.25)--(3,3) -- (3.5,2.25);
		\draw[] (3.5,2.25)--(4,3) -- (6,0);
		
		\draw[] (1,0)--(1,1);
		\draw[] (2,0)--(2,3);
		\draw[] (3,0)--(3,3);
		\draw[] (4,0)--(4,3);
		\draw[] (5,0)--(5,1);

		\draw (1,0) -- (1,1.5);
		\foreach \x in {1,2,3,4}
		\draw (1+\x/4,0) -- (1+\x/4,\x/4*1.5 +1.5);
		
		\foreach \x in {1,2,3}
		\draw (2+\x/6,0) -- (2+\x/6,3 -\x/6*1.5);
		\foreach \x in {4,5}
		\draw (2+\x/6,0) -- (2+\x/6,1.5 +\x/6*1.5);
		\foreach \x in {1,2,3,4,5}
		\draw (3-\x/6,0) -- (3-\x/6,\x/6*1.5);
		
		\foreach \x in {1,2,3}
		\draw (3+\x/6,0) -- (3+\x/6,3 -\x/6*1.5);
		\foreach \x in {4,5}
		\draw (3+\x/6,0) -- (3+\x/6,1.5 +\x/6*1.5);
		\foreach \x in {1,2,3,4,5}
		\draw (3-\x/6,0) -- (3-\x/6,\x/6*1.5);
		
		\draw (1,0) -- (1,1.5);
		\foreach \x in {1,2,3,4}
		\draw (4+\x/4,0) -- (4+\x/4, 3- \x/4*1.5);

		\draw[fill=gray] (2 ,0)--(2+1/6,0) -- (2+1/6,11/4)-- (2,3) -- cycle;
		\draw[fill=gray] (2+5/6 ,0)--(3,0)--(3,3) -- (2+5/6,11/4) -- cycle;
		\draw[fill=black] (2+2/6 ,0)--(2+4/6,0)--(2+4/6,10/4) -- (2+3/6,9/4) -- (2+2/6, 10/4) -- cycle;
		\draw[fill=gray] (3 ,0)--(3+1/6,0) -- (3+1/6,11/4)-- (3,3) -- cycle;
		\draw[fill=gray] (3+5/6 ,0)--(4,0) -- (4,3)-- (3+5/6,11/4) -- cycle;
		\draw[fill=black] (3+2/6 ,0)--(3+4/6,0)--(3+4/6,10/4) -- (3+3/6,9/4) -- (3+2/6, 10/4) -- cycle;
		
		\draw[fill=gray] (0 ,0)--(1,0)--(1,1.5) -- (0,0) -- cycle;
		\draw[fill=gray] (1 ,0)--(1+1/4,0)--(1+1/4, 1.25*1.5) -- (1,1.5) -- cycle;
		\draw[fill=gray] (1+3/4 ,0)--(2,0)-- (2,3) --  (1 +3/4, 1.75* 1.5) -- cycle;
		\draw[fill=black] (6 ,0)--(5,0)--(5, 1.5) --  cycle;
		\draw[fill=black] (4+1/4 ,0)--(4 +3/4,0)-- (4 +3/4, 15/8) --  (4 +1/4, 21/8) -- cycle;
		\end{tikzpicture}
		\caption{MWW ($\frac{23}{48},\frac{18}{48},\frac{23}{48}$).} \label{fig:M5}
	\end{subfigure} 
	
	\caption{An example showing the differences between the four different mechanisms. Agents 1, 2, and 3 receive the grey, white, and black pieces respectively. Agents' utilities appear in parentheses.}
	\label{fig:fig3}
\end{figure}

\begin{table}[h!]
	\begin{center}
		\caption{Summary of the properties of each mechanism.}
		\label{tab:properties}
	\begin{tabular}{lccccc}
		        \toprule
					&Efficient    & 	Envy-free  &	Max Utility & Dominates WW   \\
					\midrule
		Wang--Wu          & No  & Yes & No  &  \\

		Utilitarian       & Yes & No  & Yes & No \\
		LL               & Yes & No  & No  & Yes \\
		MWW  			  & No 	& Yes & No  & Yes  \\

		\bottomrule
	\end{tabular}
\end{center}
\end{table}

\begin{proposition}
	\label{thm:prop3}
	The utilitarian mechanism is Pareto optimal and maximizes the sum of agents' utilities. Leftmost leaves is Pareto optimal and Pareto dominates the original Wang--Wu mechanism. MWW is envy-free and also Pareto dominates the original Wang--Wu mechanism. 
\end{proposition}

\begin{proof}
	Proposition 2 implies that the utilitarian and leftmost leaves mechanisms are Pareto optimal. MWW is envy-free because, for any of the (up to) $3n+1$ subintervals generated, any agent values the share received by of any other agent (among those who also receive a share) as much as her own (as in Theorem 2 in \cite{wang2019cake}). Leftmost leaves Pareto dominates Wang--Wu follows because leftmost leaves is proportional, and only assigns a utility of $1/n$ to every agent when all agents' have the same preferences \citep{kyropoulou2019fair}, whereas Wang--Wu always gives a utility of $1/n$ to all agents. That MWW Pareto dominates Wang--Wu follows from the fact that, although both WW and MWW divide the cake into the same number of pieces ($3n+1$), MWW only assigns to agents cake shares of the pieces they desire.
\end{proof}

\section{Conclusion}
\label{sec:concl}
Dividing a cake fairly and efficiently is a difficult problem, even when agents have simple preferences. Nevertheless, we have proposed fair division mechanisms that are either Pareto optimal or envy-free, and which Pareto dominate the allocation generated by the envy-free procedure of Wang--Wu. Using any of these mechanisms instead is justified due to the large utility losses that the Wang--Wu mechanism can generate. 

\section*{Acknowledgements}
We are grateful to Chenhao Wang, Erel Segal-Halevi and an anonymous referee for their helpful comments. Sarah Fox proofread this paper. This work was completed while Bhavook Bhardwaj was a visiting student at Queen's University Belfast; their hospitality is gratefully acknowledged. Rajnish Kumar acknowledges financial support provided by British Council Grant UGC-UKIERI 2016-17-059.

\setlength{\bibsep}{0cm}

\end{document}